\documentclass[runningheads]{llncs}

\usepackage{amsmath,amssymb}
\usepackage{tikz}
\usepackage{bm}
\usepackage{algorithm}
\usepackage{algpseudocode}
\newcommand{\Arena}{(P,V,(V_p)_{p \in P},v_0,E)}
\newcommand{\Play}{\mathit{Play}}
\newcommand{\Hist}{\mathit{Hist}}
\renewcommand{\O}{\mathcal{O}}
\newcommand{\Inf}{\mathit{Inf}}
\newcommand{\Win}{\mathrm{Win}}
\newcommand{\G}{\mathcal{G}}
\newcommand{\bmsigma}{\bm{\sigma}}
\newcommand{\bmalpha}{\bm{\alpha}}

\newcommand{\out}{\mathrm{out}}
\newcommand{\Obj}{\mathrm{Obj}}

\newcommand{\Knw}{\mathit{knw}}
\newcommand{\PW}{\mathsf{pw}}
\newcommand{\GW}{\mathsf{gw}}
\newcommand{\PG}{\mathsf{pg}}
\newcommand{\PGW}{\mathsf{pgw}}

\newcommand{\Winnable}{\mathrm{Winnable}}

\newcommand{\Objp}{\mathrm{Obj}^{p,O_p}}
\newcommand{\Objpk}{\mathrm{Obj}^{p,O_p}_\Knw}
\newcommand{\Objpok}{\mathrm{Obj}^{p,O_p}_{\Omega,\Knw}}

\newcommand{\Out}{\mathrm{out}}
\newcommand{\F}{\mathcal{F}}
\newcommand{\Nash}{\mathrm{Nash}}

\newcommand{\calG}{\mathcal{G}}
\newcommand{\calO}{\mathcal{O}}
\newcommand{\balpha}{{\boldsymbol\alpha}}
\newcommand{\bsigma}{{\boldsymbol\sigma}}
\newcommand{\powerset}[1]{2^{#1}}

\newcommand{\Muller}{\mathop{\mathrm{Muller}}\nolimits}
\newcommand{\calF}{\mathcal{F}}

\let\emptyset=\varnothing

\newcommand{\Buchi}{\mathrm{B\ddot{u}chi}}

\title{A Game-Theoretic Approach to Indistinguishability of Winning Objectives \\as User Privacy}
\titlerunning{Indistinguishability of Winning Objectives}
\author{
	Rindo Nakanishi\inst{1} \and
	Yoshiaki Takata\inst{2} \and
	Hiroyuki Seki\inst{1}
}
\authorrunning{R. Nakanishi et al.}
\institute{
	Graduate School of Informatics, Nagoya University \\
	Furo-cho, Chikusa, Nagoya 464-8601, Japan \\
	\email{\{rindo,seki\}@sqlab.jp}
	\and School of Informatics, Kochi University of Technology \\
	Tosayamada, Kami City, Kochi 782-8502, Japan \\
	\email{takata.yoshiaki@kochi-tech.ac.jp}
}

\begin{document}
\maketitle
	\begin{abstract}
	Game theory on graphs is a basic tool in computer science.
	In this paper, we propose a new game-theoretic framework for studying the privacy protection of a user
	who interactively uses a software service.
	Our framework is based on the idea that an objective of a user using software services should
	not be known to an adversary because the objective is often closely related to personal information of the user.
	We propose two new notions, 
	$\calO$-indistinguishable strategy ($\calO$-IS) and objective-indistinguishability equilibrium (OIE).
	For a given game and a subset $\calO$ of winning objectives (or objectives in short),
	a strategy of a player is $\calO$-indistinguishable if
	an adversary cannot shrink $\calO$ by excluding any objective $O$ from $\calO$
	as an impossible objective.
	A strategy profile, which is a tuple of strategies of all players, is
	an OIE if the profile is locally maximal in the sense that
	no player can expand her set of objectives
	indistinguishable from her
	real objective from the viewpoint of an adversary.
	We show that for a given multiplayer game with Muller objectives,
	both of the existence of an $\calO$-IS and that of OIE are decidable.
\end{abstract}

\section{Introduction}
		Indistinguishability is a basic concept in security and privacy, meaning that 
anyone who does not have the access right to secret information cannot 
distinguish between a target secret data and other data. 
For example, a cryptographic protocol may be considered secure if 
the answer from an adversary who tries to attack the protocol 
is indistinguishable from a random sequence (computational indistinguishability) \cite{Go01}. 
In the database community, $k$-anonymity has been frequently used as a criterion on 
privacy of a user's record in a database; a database is $k$-anonymous if 
we cannot distinguish a target record from at least $k-1$ records whose public attribute values
are the same as those of the target record \cite{Sw02}. 

In this paper, we apply indistinguishability to defining and solving problems on 
privacy of a user who interacts with other users and/or software tools. 
Our basic framework is a multiplayer non-zero-some game played on a game arena, 
which is a finite directed graph with the initial vertex \cite{Br17,BCJ18}.  
A game has been used as the framework of reactive synthesis problem \cite{PR89,FKL10}.
A play in a game arena is an infinite string of vertices starting with the initial vertex and 
along edges in the game arena.  
To determine the result (or payoff) of a play, 
a winning objective $O_p$ is specified for each player $p$. 
If the play satisfies $O_p$, then we say that the player $p$ wins in this play. 
Otherwise, the player $p$ loses. 
A play is determined when each player determines her strategy in the game. 
A strategy $\sigma$ of a player $p$ is called a winning strategy if the player $p$ always wins
by using $\sigma$, i.e., 
any play consistent with the strategy $\sigma$ satisfies her winning objective 
regardless of the other players' strategies. 
One of the main concerns in game theory is to decide whether there is a winning strategy
for a given player $p$ and if so, to construct a winning strategy for $p$. 
Note that there may be more than one winning strategies for a player; 
she can choose any one among such winning strategies. 
In the literatures, a winning objective is {\em a priori} given as a component of a game. 
In this study, we regard a winning objective of a player is her private information
because objectives of a user of software services are closely related to her private information.
For example, 
users of e-commerce websites may select products to purchase 
depending on their preference, income and health condition, etc., 
which are related to private information of the users. 
Hence, it is natural for a player to choose a winning strategy that maximizes 
the indistinguishability of her winning objective from the viewpoint of an adversary who may 
observe the play and recognize which players win the game. 
For a subset $\cal{O}$ of winning objectives which a player $p$ wants to be indistinguishable from one another, 
we say that a strategy of $p$ is $\cal{O}$-indistinguishable if an adversary cannot make $\cal{O}$
smaller as the candidate set of winning objectives. 
The paper discusses the decidability of some problems related to $\cal{O}$-indistinguishability. 

Another important problem in game theory is to find a good combination of strategies of all players, which 
provides a locally optimal play.  
A well-known criterion is Nash equilibrium.  A combination of strategies (called a strategy profile)
is a Nash equilibrium if any player losing the game in that strategy profile 
cannot make herself a winner by changing her strategy alone. 
This paper introduces objective-indistinguishability equilibrium (OIE) as a criterion of local optimality 
of a strategy profile; a strategy profile is OIE if and only if no player can extend 
the indistinguishable set of winning objectives by changing her strategy alone. 
The paper also provides the decidability results on OIE. 

\paragraph*{Related work}
As already mentioned, this paper focuses on multiplayer turn-based non-zero-sum games. 
There is a generalization of games where each player can only know 
partial information on the game, which is called an imperfect information game\cite{AG22,BMMRY21,BMV17,CD10,CHP14}. 
While the indistinguishability proposed in this paper shares such restricted observation
with imperfect information games, the large difference is that 
we consider an adversary who is not a player but an individual who observes partial information
on the game while players themselves may obtain only partial information in imperfect information games. 


There are many privacy notions and a vast amount of literatures studying privacy issues. 
Among them, $k$-anonymity is one of the well-known notions originated in the database community. 
A database $D$ is $k$-{\em anonymous} \cite{Sa01,Sw02} if for any record $r$ in $D$, there are at 
least $k-1$ records different from $r$ such that the values of quasi-identifiers
of $r$ and these records are the same.  Here, a set of quasi-identifiers is a subset of 
attributes that can `almost' identify the record such as $\{$zip-code, birthday, income$\}$. 
Hence, if $D$ is $k$-anonymous, an adversary knowing the quasi-identifiers of some user $u$ 
cannot identify the record of $u$ in $D$ among the $k$ records with the same values of the quasi-identifiers.
Methods for transforming a database to the one satisfying $k$-anonymity have been 
investigated \cite{BKBL07,BA05}. 
Also, refined notions such as $\ell$-{\em diversity} \cite{MGK07} and $t$-{\em closeness} \cite{LLV07} 
have been proposed by considering the statistical distribution of the attribute values. 

However, these notions suffer from so called non-structured zero and mosaic effect. 
Actually, it is known that there is no way of protecting perfect privacy from an 
adversary who can use an arbitrary external information except the target privacy itself. 
The notion of $\varepsilon$-differential privacy 
where $\varepsilon>0$ was proposed to overcome the weakness of the classical notions of privacy. 
In a nutshell, a query $Q$ to a database $D$ is $\varepsilon$-\emph{differentially private}
(abbreviated as $\varepsilon$-DP) \cite{DMNS06,Dw06} 
if for any person $u$, 
the probability that we can infer whether the information on $u$ is contained in $D$ or not  
by observing the result of $Q(D)$ is negligible (very small) in terms of $\varepsilon$.  
(Also see \cite{Dw08,DR14} as comprehensive tutorials.) 
As the privacy protection of individual information used in data mining and machine learning
is becoming a serious social problem (see \cite{SSSS17} for example), 
methods of data publishing that guarantees $\varepsilon$-DP have been extensively studied
\cite{FWCY10,ABCP13,ACGMMTZ16,SS15,SSSS17}.

Quantitative information flow (abbreviated as QIF) \cite{CPP08,Sm09} is 
another way of formalizing privacy protection or information leakage.
QIF of a program $P$ is the mutual information of the secret input 
$X$ and the public output $Y$ of the program $P$ in the sense of Shannon theory where 
the channel between $X$ and $Y$ is a program which has logical semantics. 
Hence, QIF analysis uses not only the calculation of probabilities but also 
program analysis such as type inference \cite{CHM07} and symbolic execution. 

We have mentioned a few well-known approaches to formally modeling privacy protection 
in software systems; however, these privacy notions, 
even QIF that is based on the logical semantics of a program, 
share the assumption that private information is a static value or a distribution of values.

In contrast, our approach assumes that privacy is 
a purpose of a user's behavior.
The protection of this kind of privacy has not been studied to the best of our knowledge.  
As an extension of rational synthesis, 
Kupferman and Leshkowitz have introduced the synthesis problem of privacy preserving systems~\cite{KL22}; 
the problem is for given multivalued LTL formulas representing secrets as well as an LTL formula representing a specification, 
to decide whether there is a reactive program that satisfies the specification 
while keeping the values of the formulas representing secrets unknown. 
This study treats the secrets as values as in the previous studies, and the approach is very different from ours.

While we adopt Nash equilibrium, there are other criteria for local optimality of strategy profiles, namely, 
secure equilibrium (SE) \cite{CHJ06} and doomsday equilibrium (DE) \cite{CDFR17}. 
SE is a strategy profile such that
no player can improve her payoff or punish any other player without loss of her own payoff
by changing only her strategy.
SE is used for a verification of component-based systems where each component has its own objective.
DE is a strategy profile such that all players are winners and
each player can make all players lose as retaliation when she becomes a loser because some other players change their strategies.
SE and DE are secure in the sense that no player is punished by other player(s)
and not directly related to user privacy.

\paragraph*{Outline}
In Section \ref{sec:prel}, we define some notions and notations 
on multiplayer turn-based deterministic games used in subsequent sections.
Moreover, in Section \ref{sec:prel}, we define an $(\bmalpha_1,\ldots,\bmalpha_n)$-Nash equilibrium (NE) 
as a strategy profile which is simultaneously a NE 
for all objective profiles $\bmalpha_1,\ldots,\bmalpha_n$.
We show that whether there exists an $(\bmalpha_1,\ldots,\bmalpha_n)$-NE is decidable in Theorem 2,
which will be used in Section \ref{sec:results}.
In Section \ref{sec:concept}, we propose two new notions, namely
$\calO$-indistinguishable strategy ($\calO$-IS) and objective-indistinguishability equilibrium (OIE).
$\calO$-IS is a strategy such that an adversary cannot shrink the set $\calO$ of candidate objectives of a player.
OIE is a strategy profile such that no player can expand her own set of candidate objectives.
In Section \ref{sec:results}, we show that for a given multiplayer game with Muller objectives,
both the existence of an $\calO$-IS and that of OIE are decidable.
In Section \ref{sec:conclusion}, we give a conclusion of this paper.

\section{Preliminaries}\label{sec:prel}
		\begin{definition}
	A game arena is a tuple $G = \Arena$, where
	\begin{itemize}
		\item $P$ is a finite set of players,
		\item $V$ is a finite set of vertices,
		\item $(V_p)_{p \in P}$ is a partition of $V$, 
			  namely, $V_i \cap V_j = \varnothing$ for all $i\neq j \ (i,j \in P)$ and
			  $\bigcup_{p \in P} V_p = V$,
	  	\item $v_0 \in V$ is the initial vertex, and
		\item $E \subseteq V \times V$ is a set of edges.
	\end{itemize}
\end{definition}
As defined later, a vertex in $V_p$ is controlled by a player $p$, i.e., 
when a play is at a vertex in $V_p$, the next vertex is selected by player $p$.
This type of games is called \emph{turn-based}.
There are other types of games, concurrent and stochastic games.
In a concurrent game~\cite{AG22}, each vertex may be controlled by more than one (or all) players.
In a stochastic game~\cite{Um08,UW11,CAH05}, each vertex is controlled by a player or a special entity \emph{nature}
who selects next nodes according to a probabilistic distribution for next nodes given as a 
part of a game arena.
Moreover, a strategy of a player selects a next node stochastically.
In this paper, we consider only deterministic turn-based games.

	\paragraph*{Play and history}
		An infinite string of vertices $v_0 v_1 v_2 \cdots \ (v_i \in V, i \geq 0)$ starting from 
the initial vertex $v_0$ is a \emph{play}
if $(v_i,v_{i+1}) \in E$ for all $i \geq 0$.
A \emph{history} is a non-empty (finite) prefix of a play.
The set of all plays is denoted by $\Play$ and the set of all histories is denoted by $\Hist$.
We often write a history as $hv$ where
$h \in \Hist \cup \{\varepsilon\}$ and $v \in V$.
For a player $p \in P$, let $\Hist_p = \{ hv \in \Hist \mid v \in V_p \}$.
That is, $\Hist_p$ is the set of histories ending with a vertex controlled by player $p$.
For a play $\rho=v_0v_1v_2\cdots \in \Play$, we define $\Inf(\rho) = \{ v\in V\mid \forall i\geq0.\ \exists j\geq i.\ v_j=v\}$.

	\paragraph*{Strategy}
		For a player $p \in P$, a \emph{strategy} of $p$ is a function $\sigma_p: \Hist_p \to V$
such that $(v, \sigma_p(hv)) \in E$ for all $hv \in \Hist_p$.
At a vertex $v\in V_p$, player $p$ chooses $\sigma_p(hv)$ as the next vertex according to her strategy $\sigma_p$. 
Note that because the domain of $\sigma_p$ is $Hist_p$, the next vertex may depend on the whole history in general.
Let $\Sigma^p_{\G}$ denote the set of all strategies of $p$.
A \emph{strategy profile} is a tuple $\bmsigma = (\sigma_p)_{p \in P}$
of strategies of all players, namely $\sigma_p \in \Sigma^p_{\G}$ for all $p \in P$.
Let $\Sigma_{\G}$ denote the set of all strategy profiles.
For a strategy profile $\bmsigma \in \Sigma_{\G}$ and
a strategy $\sigma'_p \in \Sigma^p_{\G}$ of a player $p\in P$,
let $\bmsigma[p \mapsto \sigma'_p]$ denote the strategy profile obtained from $\bmsigma$ by replacing 
the strategy of $p$ in $\bmsigma$ with $\sigma'_p$.
We define the function $\out_{\G}:\Sigma_{\G} \to \Play$ as $\out_{\G}((\sigma_p)_{p \in P}) = v_0 v_1 v_2 \cdots$
where $v_{i+1} = \sigma_p(v_0\cdots v_i)$ for all $i \geq 0$ and for $p\in P$ with $v_i \in V_p$.
We call the play $\out_\G(\bmsigma)$ the \emph{outcome} of $\bmsigma$.
We also define the function $\out^p_{\G}:\Sigma^p \to 2^{\Play}$ for each $p \in P$ as
$\out^p_{\G}(\sigma_p) = 
	\{ v_0 v_1 v_2 \cdots \in \Play \mid 
	\text{$v_i \in V_p \Rightarrow v_{i+1}=\sigma_p(v_0\cdots v_i)$ for all $i\geq 0$}\}$.
A play $\rho \in \out^p_\G(\sigma_p)$ is called a play consistent with the strategy $\sigma_p$ of player $p$. 
By definition, for a strategy profile $\bmsigma = (\sigma_p)_{p\in P} \in \Sigma_{\G}$,
it holds that $\bigcap_{p\in P}\out^p_{\G}(\sigma_p)=\{\out_{\G}(\bmsigma)\}$.

	\paragraph*{Objective}
In this paper, we assume that the result that a player obtains from a play is either a winning or a losing.
Since we are considering non-zero-sum games, one player's winning does not mean other players' losing.
Each player has her own winning condition over plays, and we model the condition as a subset $O$ of plays; 
i.e., the player wins if the play belongs to the subset $O$.
We call the subset $O\subseteq\Play$ the \emph{objective} of that player.
In this paper, we focus on the following important classes of objectives:
\begin{definition}\label{def:obj}
Let $U\subseteq V$ be a subset of vertices, 
$c:V\to\mathbb{N}$ be a coloring function,
$(F_k,G_k)_{1\leq k\leq l}$ be pairs of sets $F_k,G_k\subseteq V$ and
$\calF\subseteq 2^V$ be a subset of subsets of vertices.
We will use $U$, $c$, $(F_k,G_k)_{1\leq k\leq l}$ and $\calF$ as finite representations for specifying an objective as follows:
\begin{itemize}
	\item B\"{u}chi objective: $\mathrm{B\ddot{u}chi}(U)=\{ \rho\in\Play \mid \Inf(\rho) \cap U \neq \varnothing\}$.
	\item Co-B\"{u}chi objective: $\mathrm{Co\mathchar`-B\ddot{u}chi}(U)=\{ \rho\in\Play \mid \Inf(\rho) \cap U = \varnothing\}$.
	\item Parity objective: $\mathrm{Parity}(c) = \{ \rho=v_0v_1v_2\cdots\in\Play \mid \text{$\max(\{c(v_i) \mid i\geq 0\})$ is even} \}$.
	\item Rabin objective: $\mathrm{Rabin}\bigl((F_k,G_k)_{1\leq k\leq l}\bigr)=
		\{\rho\in\Play \mid 1\leq \exists k\leq l.\ \Inf(\rho)\cap F_k=\varnothing \wedge \Inf(\rho) \cap G_k \neq \varnothing \}$.
	\item Streett objective: $\mathrm{Streett}\bigl((F_k,G_k)_{1\leq k\leq l}\bigr)=
		\{\rho\in\Play \mid 1\leq \forall k\leq l.\ \Inf(\rho)\cap F_k\neq\varnothing \vee \Inf(\rho) \cap G_k=\varnothing \}$.
	\item Muller objective: $\mathrm{Muller}(\F)=\{ \rho \in \Play \mid \Inf(\rho)\in\F\}$.
\end{itemize}
\end{definition}
Note that each objective defined in Definition \ref{def:obj} is also a Muller objective:
For example, $\Buchi(U)=\Muller(\{I\subseteq V\mid I\cap U \neq \varnothing\})$.
We define the description length of a Muller objective
$\Muller(\calF)$ for $\calF\subseteq\powerset{V}$ is $|V|\cdot|\calF|$,
because each element of $\calF$, which is a subset of $V$,
can be represented by a bit vector of length $|V|$.
By $\Omega \subseteq 2^{\Play}$, we refer to a certain class of objectives.
For example, $\Omega = \{ \text{B\"{u}chi}(U) \mid U \subseteq V \} \subseteq 2^{Play}$
is the class of B\"{u}chi objectives.

An \emph{objective profile} is a tuple $\bmalpha = (O_p)_{p \in P}$ of objectives of all players,
namely $O_p \subseteq \Play$ for all $p \in P$.
For a strategy profile $\bmsigma \in \Sigma_{\G}$ and an objective profile $\bmalpha = (O_p)_{p \in P}$,
we define the set $\Win_\G(\bmsigma,\bmalpha) \subseteq P$ of winners as
$\Win_\G(\bmsigma,\bmalpha) = \{ p \in P \mid \out_{\G}(\bmsigma) \in O_p \}$.
That is, a player $p$ is a winner if and only if  $\out_\G(\bmsigma)$ belongs to the objective $O_p$ of $p$. 
If $p \in \Win_\G(\bmsigma, \bmalpha)$, we also say that $p$ wins the game $\G$ with $\bmalpha$
(by the strategy profile $\bmsigma$). 
Note that it is possible that there is no player who wins the game or all the players win the game. 
In this sense, a game is non-zero-sum. 

We abbreviate $\Sigma^p_{\G},\Sigma_{\G},\out^p_{\G},\out_{\G}$ and $\Win_{\G}$ as
$\Sigma^p,\Sigma,\out^p,\out$ and $\Win$, respectively, if $\G$ is clear from the context.

	\paragraph*{Winning strategy}
		For a game arena $\G$, a player $p \in P$ and an objective $O_p \subseteq \Play$,
a strategy $\sigma_p \in \Sigma^p$ of $p$ such that $\out^p(\sigma_p)\subseteq O_p$ is called a \emph{winning strategy}
of $p$ for $\G$ and $O_p$
because if $p$ takes $\sigma_p$ as her strategy then she wins against any combination of strategies of the other players.
(Recall that $\out^p(\sigma_p)$ is the set of all plays consistent with $\sigma_p$.)
For a game arena $\G$ and a player $p\in P$,
we define the set $\Winnable^p_\G$ of objectives permitting a winning strategy as 
$\Winnable^p_\G=\{O \mid \exists \sigma_p \in \Sigma^p_\G.\ \out^p_\G(\sigma_p)\subseteq O\}$.
For a player $p$, $O\in\Winnable^p_\G$ means that $p$ has a winning strategy for $\G$ and $O$. 
On the existence of a winning strategy for a Muller objective, the following theorem is known.
\begin{theorem}[{\cite[Theorem 21]{Br17}}]\label{thm:Muller}
Let $\G=\Arena$ be a game arena and $O_p\subseteq \Play$ be a Muller objective of $p\in P$.
Deciding whether there exists a winning strategy of $p$ for $O_p$ is $\mathsf{P}$-complete.
\qed
\end{theorem}

For such non-zero-sum multiplayer games as considered in this paper, 
we often use Nash equilibrium, defined below, as a criterion for a strategy profile to be locally optimal.

	\paragraph*{Nash equilibrium}
		Let $\bmsigma\in \Sigma$ be a strategy profile and $\bmalpha=(O_p)_{p\in P}$ be an objective profile.
A strategy profile $\bmsigma$ is called a \emph{Nash equilibrium} (NE) for $\bmalpha$ if it holds that
\[
	\forall p \in P.\ \forall \sigma_p \in \Sigma^p.\ p\in\Win(\bmsigma[p\mapsto\sigma_p],\bmalpha) 
	\Rightarrow p \in\Win(\bmsigma,\bmalpha).
\]
Intuitively, $\bsigma$ is a NE if every player $p$ cannot improve the result (from losing to winning) 
by changing her strategy alone.
For a strategy profile $\bmsigma\in\Sigma$,
we call a strategy $\sigma_p\in\Sigma^p$ 
such that $p\notin\Win(\bmsigma,\bmalpha)\wedge p\in\Win(\bmsigma[p\mapsto \sigma_p],\bmalpha)$
a \emph{profitable deviation} of $p$ from $\bmsigma$.
Hence, $\bmsigma$ is a NE if and only if no player has a profitable deviation from $\bmsigma$.
Because $p\in \Win(\bmsigma,\bmalpha)$ is equivalent to $\out(\bmsigma) \in O_p$, 
a strategy profile $\bmsigma \in \Sigma$ is a NE for $\bmalpha$ if and only if 
\begin{equation}
	\forall p\in P.\ \forall \sigma_p\in\Sigma^p.\ \out(\bmsigma[p\mapsto\sigma_p])\in O_p
	\Rightarrow \out(\bmsigma)\in O_p.
	\label{con:nash}
\end{equation}
We write Condition (\ref{con:nash}) as $\Nash(\bmsigma,\bmalpha)$.

Below we define an extension of NE that is a single strategy profile 
simultaneously satisfying the condition of NE for more than one objective profiles. 
We can prove that the existence of this extended NE is decidable (Theorem \ref{thm:n-nash}), 
and later we will reduce some problems to the existence checking of this type of NE\@.
\begin{definition}\label{def:n-nash}
	For a game arena $\G=\Arena$ and objective profiles $\bmalpha_1,\ldots,\bmalpha_n$,
	a strategy profile $\bmsigma\in\Sigma$ is called an \emph{$(\bmalpha_1,\ldots,\bmalpha_n)$-Nash equilibrium}
	if $\Nash(\bmsigma,\bmalpha_j)$ for all $1\leq j\leq n$.
\end{definition}

\begin{theorem}\label{thm:n-nash}
	Let $\G=\Arena$ be a game arena and $\bmalpha_j=(O_p^j)_{p\in P}$ $(1\leq j \leq n)$ be objective profiles
	over Muller objectives.
	Deciding whether there exists an $(\bmalpha_1,\ldots,\bmalpha_n)$-NE
	is decidable.
\end{theorem}
A proof of this theorem is given in the Appendix.

\section{Indistinguishable Strategy and Related Equilibrium}\label{sec:concept}
	In this section, we propose two new notions concerning on the privacy of a player: 
indistinguishable strategy and objective-indistinguishability equilibrium. 
We first define the set of possible objectives of a player in the viewpoint of an adversary 
that can observe restricted information on a game, a play and its result (i.e., which players win).

We assume that an adversary guesses objectives of players from the three types of information:
a play ($\mathsf{p}$), a game arena ($\mathsf{g}$) and a set of winners ($\mathsf{w}$) of the play.
We use a word $\Knw \in \{\PW,\GW,\PG,\PGW\}$ to represent a type of information that an adversary can use.
For example, an adversary guesses objectives from a play and winners when $\Knw=\PW$.
We do not consider the cases where $\Knw$ is a singleton because an adversary cannot guess anything from such information.
In either case, we implicitly assume that an adversary knows the set $V$ of vertices of the game arena.
Let $p \in P$ be a player, $O_p \subseteq \Play$ be an objective of $p$ 
and $\Omega \subseteq 2^\Play$ be one of the classes of objectives.
We define the function $\Objpok:\Sigma \to 2^{\Omega}$ as follows, which 
maps a strategy profile $\bmsigma\in\Sigma$ to the set of objectives of $p$ that an adversary guesses:
\begin{align*}
  \Objp_{\Omega,\PW}(\bsigma) &= \{ O\subseteq V^\omega\mid
     \begin{aligned}[t]
       &(\Out(\bsigma)\in O \land p\in\Win(\bsigma,\bmalpha))\lor {} \\
       &(\Out(\bsigma)\notin O \land p\notin\Win(\bsigma,\bmalpha)) \},
     \end{aligned}
  \\
  \Obj^{p,O_p}_{\Omega,\GW}(\bsigma) &= \{ O\in\Omega\mid
     \begin{aligned}[t]
	   &(p\in\Win(\bsigma,\bmalpha) \land O\neq\varnothing) \lor {}\\
	   &(p\notin\Win(\bsigma,\bmalpha) \land O\notin\Winnable^p) \},
     \end{aligned}
  \\
  \Obj^{p,O_p}_{\Omega,\PG}(\bsigma) &= \{ O\in\Omega\mid
       \Out(\bsigma)\in O
       \lor (\Out(\bsigma)\notin O \land O\notin\Winnable^p) \},
  \\
  \Obj^{p,O_p}_{\Omega,\PGW}(\bsigma) &= \{ O\in\Omega\mid
     \begin{aligned}[t]
       &(\Out(\bsigma)\in O \land p\in\Win(\bsigma,\bmalpha)) \lor {} \\
       &(\Out(\bsigma)\notin O \land p\notin\Win(\bsigma,\bmalpha) \land
         O\notin\Winnable^p) \},
     \end{aligned}
\end{align*}
where $\bmalpha$ is any objective profile in which the objective of $p$ is $O_p$.
(Note that for a given $\bmsigma$ whether $p\in \Win(\bmsigma, \bmalpha)$ or not 
does not depend on objectives of the players other than $p$ 
and hence we can use an arbitrary $\bmalpha$ containing $O_p$.)

The definitions of $\Objpok$ are based on the following ideas.
When $\Knw=\PW$, we assume that an adversary can observe the play and the set of winners
but he does not know the game arena.
The adversary can infer that 
the play $\out(\bmsigma)$ he observed belongs to the objective of a player $p$ 
if the adversary knows that $p$ is a winner,
and $\out(\bmsigma)$ does not belong to the objective of $p$ if $p$ is not a winner.
Note that the adversary does not know the real objective $O_p$ of player $p$.
For the adversary, any $O \subseteq V^\omega$ satisfying $\out(\bmsigma) \in O$ 
is a candidate of the objective of player $p$ when $p$ is a winner.
Similarly, any $O\subseteq V^\omega$ satisfying $\out(\bmsigma) \not\in O$ 
is a candidate objective of $p$ when $p$ is not a winner.
An adversary does not know the game arena because $\Knw=\PW$, that is,
he does not know the set of edges in the arena.
Therefore, the candidate objective $O$ cannot be restricted to a subset of plays 
(i.e., infinite strings of vertices along the edges in the game arena), 
but $O$ can be an arbitrary set of infinite strings of the vertices consistent with the information obtained by the adversary.

When $\Knw=\GW$, an adversary cannot observe the play, but he knows the game arena and can observe the set of winners.
If $p$ is a winner, the adversary can infer that 
$p$ has a strategy $\sigma_p$ such that $\out^p(\sigma_p) \cap O_p \neq \varnothing$.
Because there exists such a strategy $\sigma_p$ for all $O_p$ other than $\varnothing$,
he can remove only $\varnothing$ from the set of candidates for $p$'s objective.
On the other hand, if $p$ is a loser, 
the adversary can infer that $p$ has no winning strategy for $O_p$ because
we assume that every player takes a winning strategy for her objective when one exists.
Therefore, when $p$ loses, the adversary can narrow down the set of candidates for $p$'s objective 
to the set of objectives without a winning strategy.

The definition where $\Knw=\PG$ can be interpreted in a similar way.
Note that we have $\Objp_{\Omega,\PGW}(\bmsigma) = \Objp_{\Omega,\PW}(\bmsigma) \cap \Objp_{\Omega,\GW} \cap \Objp_{\Omega,\PG}$.

Since $p\in \Win(\bmsigma, \bmalpha)$ is equivalent to $\out(\bmsigma)\in O_p$, 
the above definitions can be rephrased as follows:
\begin{align*}
  \Obj^{p,O_p}_{\Omega,\PW}(\bsigma) &= \{ O\subseteq V^\omega\mid
    \Out(\bsigma)\in (O\cap O_p)\cup(\overline{O}\cap\overline{O_p}) \},
  \\
  \Obj^{p,O_p}_{\Omega,\GW}(\bsigma) &= \{ O\in\Omega\mid
     \begin{aligned}[t]
		 &(O\in\Winnable^p \Rightarrow \Out(\bsigma)\in O_p) \land {} \\
		 &(O=\varnothing \Rightarrow \out(\bmsigma)\notin O_p)\},
     \end{aligned}
  \\
  \Obj^{p,O_p}_{\Omega,\PG}(\bsigma) &= \{ O\in\Omega\mid
       O\in\Winnable^p \Rightarrow \Out(\bsigma)\in O \},
  \\
  \Obj^{p,O_p}_{\Omega,\PGW}(\bsigma) &= \{ O\in\Omega\mid
     \begin{aligned}[t]
       &\Out(\bsigma)\in (O \cap O_p)\cup(\overline{O}\cap\overline{O_p})
        \land {} \\
       &(O\in\Winnable^p \Rightarrow \Out(\bsigma)\in O \cap O_p) \}.
     \end{aligned}
\end{align*}

The reader may wonder why $O_p$ appears in this (alternative) definition 
in spite of the assumption that the adversary does not know $O_p$.
The condition $\out(\bmsigma) \in O_p$ (or $\not\in O_p$) only means that 
the adversary knows whether $p$ is a winner (or a loser) without knowing $O_p$ itself.
\begin{example}\label{exm:obj}
Figure \ref{fig:obj} shows a $1$-player game arena 
$\G=(\{1\},V,(V),v_0,E)$ where
$V=\{v_0,v_1,v_2\}$ and $E=\{(v_0,v_1),(v_0,v_2),(v_1,v_1),(v_2,v_2)\}$.
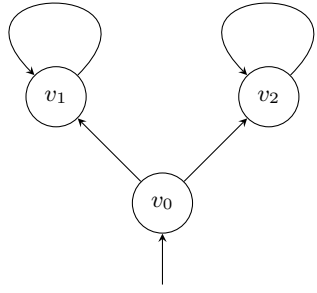
\begin{figure}[t]
        \centering
            \begin{tikzpicture}[everynode/.style={circle,draw,minimum size=0.8cm},>=stealth,node distance=2cm]
            \node[everynode] (v0) {$v_0$};
            \node[below of=v0,node distance=1.2cm] (start) {};
            \node[everynode,above left of=v0] (v1) {$v_1$};
            \node[everynode,above right of=v0] (v2) {$v_2$};
            \draw[->,loop] (v1) to (v1);
            \draw[->,loop] (v2) to (v2);
            \draw[->] (v0) to (v1);
            \draw[->] (v0) to (v2);
            \draw[->] (start) to (v0);
        \end{tikzpicture}
            \caption{$1$-player game arena with B\"{u}chi objectives}
            \label{fig:obj}
\end{figure}
We specify a B\"{u}chi objective by a set of accepting states, e.g.,
let $\langle v_1 \rangle$ denote $\Buchi(\{v_1\})=\{ \rho\in V^\omega \mid \Inf(\rho) \cap \{v_1\} \neq \varnothing\}$.
In this example, we assume the objective of player $1$ is $\langle \rangle=\varnothing\subseteq\Play$.
Therefore, player $1$ always loses regardless of her strategy.
There are only two strategies $\sigma_1$ and $\sigma_2$ of player $1$.
The strategy $\sigma_1$ takes the vertex $v_1$ as the next vertex at the initial vertex $v_0$
and then keeps looping in $v_1$.
On the other hand, the strategy $\sigma_2$ takes $v_2$ at $v_0$ and then keeps looping in $v_2$.
Let $\sigma_1$ be the strategy player $1$ chooses.
We have the play $\rho=\out(\sigma_1) =v_0v_1v_1v_1\cdots$.

We assume that an adversary knows that the objective of player $1$ is a B\"{u}chi objective.
Then, for each type of information $\Knw\in\{\PW,\GW,\PG,\PGW\}$, 
$\Obj^{1,\varnothing}_{\Buchi,\Knw}(\sigma_1)$ becomes as follows:
\begin{itemize}
	\item If $\Knw=\PW$, 
		then an adversary can deduce that $v_1$ is not an accepting state because he knows that 
		$\Inf(v_0v_1v_1\cdots) = \{ v_1 \}$ and player $1$ loses.
		Therefore, we have 
		$\Obj^{1,\varnothing}_{\Buchi,\PW}(\sigma_1)=
			\{\langle\rangle,\langle v_0 \rangle, \langle v_2 \rangle, \langle v_0,v_2 \rangle\}$.
		Note that in this game arena, there is no play passing $v_0$ infinitely often, 
		and thus $\langle \rangle$ and $\langle v_0 \rangle$ 
		(resp. $\langle v_2 \rangle$ and $\langle v_0,v_2 \rangle$) are equivalent actually. 
		However, because an adversary does not know the game arena when $\Knw = \PW$, 
		he should consider every infinite string over $V$ would be a play and thus 
		$\langle \rangle$ and $\langle v_0 \rangle$ are different for him when $\Knw = \PW$.  
		In the other cases where an adversary knows the game arena, 
		he also knows e.g. $\langle \rangle$ and $\langle v_0 \rangle$ are equivalent 
		and thus he would consider $\Omega = \{\langle \rangle,\langle v_1 \rangle,\langle v_2 \rangle,\langle v_1,v_2\rangle\}$.
	\item If $\Knw=\GW$, 
		then an adversary can deduce that neither $v_1$ nor $v_2$ is an accepting state because
		player $1$ loses in spite of the fact that there are strategies that pass through $v_1$ or $v_2$ infinitely often.
		Therefore, $\Obj^{1,\varnothing}_{\Buchi,\GW}(\sigma_1)=\{\langle\rangle\}$.
		That is, an adversary can infer the complete information.
	\item If $\Knw=\PG$,
		then an adversary can deduce that $\langle v_2 \rangle$ does not belong to 
		$\Obj^{1,\varnothing}_{\Buchi,\PG}(\sigma_1)$ because player $1$ did not take $\sigma_2$
		to pass through $v_2$ infinitely often.
		That is, if $\langle v_2 \rangle$ were the objective of player $1$, 
		then it meant she chose losing strategy $\sigma_1$ instead of winning strategy $\sigma_2$, 
		which is unlikely to happen.
		Therefore, we have 
		$\Obj^{1,\varnothing}_{\Buchi,\PG}(\sigma_1)=
			\{\langle \rangle, \langle v_1 \rangle, \langle v_1,v_2 \rangle\}$.
	\item If $\Knw=\PGW$,
		we have 
		\[
			\Obj^{1,\varnothing}_{\Buchi,\PGW}(\sigma_1)=
			\bigcap_{\Knw\in\{\PW,\GW,\PG\}}\Obj^{1,\varnothing}_{\Buchi,\Knw}(\sigma_1)
			= \{ \langle \rangle \}.
		\]
\end{itemize}
\end{example}

	\paragraph*{$\O$-indistinguishable strategy}
	\begin{definition}
	Let $\G=\Arena$ be a game arena, $\sigma_p \in \Sigma^p$ be a strategy of $p\in P$, 
	$\Omega\subseteq2^\Play$ be one of the classes of objectives defined in Definition \ref{def:obj},
	$O_p \in \Omega$ be an objective of $p$ 
	and $\Knw \in \{\PW,\GW,\PG,\PGW\}$ be a type of information that an adversary can use.
	For any set $\calO\subseteq2^\Play$ of objectives such that 
	$\calO \subseteq \bigcap_{\bmsigma \in \Sigma}\Objpok(\bmsigma[p\mapsto \sigma_p])$,
	we call $\sigma_p$ an \emph{$\calO$-indistinguishable strategy} ($\calO$-IS) of $p$ (for $O_p$ and $\Knw$).
\end{definition}
Intuitively, when a player takes an $\calO$-IS as her strategy, 
an adversary cannot narrow down the set of candidates of $p$'s objective from $\calO$
by the following reason.
By definition, any objective $O$ belonging to $\calO$ also belongs to $\Objpok(\bmsigma[p \mapsto \sigma_p])$ 
for the combination of $\sigma_p$ and any strategies of the players other than $p$.
This means that such an objective $O$ is possible as the objective of $p$ 
from the viewpoint of the adversary who can use a type of information specified by $\Knw$.
If an $\calO$-IS $\sigma_p\in\Sigma^p$ is a winning strategy of $p$, 
then we call $\sigma_p$ a \emph{winning $\calO$-IS} of $p$.

\begin{example}
	Figure \ref{fig:O-IS} shows a $1$-player game arena $\G=(\{1\},V,(V),v_0,E)$ where
	$V=\{v_0,v_1,v_2\}$ and $E= \{ (v_0,v_0),(v_0,v_1),(v_1,v_0),(v_1,v_2),(v_2,v_0)  \}$.
	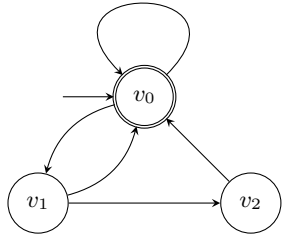
\begin{figure}[t]
		\centering
		\begin{tikzpicture}[everynode/.style={circle,draw,minimum size=0.8cm},>=stealth,node distance=2cm]
			\node[everynode,double] (v0) {$v_0$};
			\node[left of=v0,node distance=1.2cm] (start) {};
			\node[everynode,below left of=v0] (v1) {$v_1$};
			\node[everynode,below right of=v0] (v2) {$v_2$};
			\draw[->,loop] (v0) to (v0);
			\draw[->] (v0) to [bend right](v1);
			\draw[->] (v1) to [bend right] (v0);
			\draw[->] (v1) to (v2);
			\draw[->] (v2) to (v0);
			\draw[->] (start) to (v0);
		\end{tikzpicture}
			\caption{$1$-player game arena with B\"{u}chi objectives}
			\label{fig:O-IS}
	\end{figure}
	We use the same notation of B\"{u}chi objectives as Example \ref{exm:obj}, and
	in this example the objective of player $1$ is $\langle v_0 \rangle\subseteq\Play$.
	We assume that an adversary knows that the objective of player $1$ is a B\"{u}chi objective.
	In this example, we focus on $\Knw=\PW$.
	We examine the following three strategies of player $1$, all of which result in player $1$'s winning.
	\begin{itemize}
		\item Let $\sigma_1\in\Sigma^1$ be a strategy of player $1$ such that $\out(\sigma_1)=v_0v_0v_0\cdots$.
			Since player $1$ wins, an adversary can deduce that $v_0$ must be an accepting state. 
			Therefore,
			$\Obj^{1,\langle v_0 \rangle}_{\Buchi,\PW}(\sigma_1)
			  =\{\langle v_0 \rangle,\langle v_0,v_1\rangle,\langle v_0,v_2 \rangle,\langle v_0,v_1,v_2 \rangle\}$.
			For all $\calO\subseteq\Obj^{1,\langle v_0 \rangle}_{\Buchi,\PW}(\sigma_1)$,
			$\sigma_1$ is an $\calO$-IS (for $\langle v_0 \rangle$ and $\Knw=\PW$).
		\item Let $\sigma_2\in\Sigma^1$ be a strategy of player $1$ such that $\out(\sigma_1)=v_0v_1v_0v_1\cdots$.
			In a similar way as the above case, an adversary can deduce that $v_0$ or $v_1$ (or both) must be an accepting state. 
			Therefore,
			$
				\Obj^{1,\langle v_0 \rangle}_{\Buchi,\PW}(\sigma_2)
			  =\{\langle v_0 \rangle$, $\langle v_1 \rangle$,
			   $\langle v_0,v_1\rangle$, $\langle v_1,v_2 \rangle$, $\langle v_2,v_0 \rangle,
			   \langle v_0,v_1,v_2 \rangle\}.
		   $
			For all $\calO\subseteq\Obj^{1,\langle v_0 \rangle}_{\Buchi,\PW}(\sigma_2)$,
			$\sigma_2$ is an $\calO$-IS. 
		\item Let $\sigma_3\in\Sigma^1$ be a strategy of player $1$ such that $\out(\sigma_3)=v_0v_1v_2v_0v_1v_2\cdots$.
			In a similar way as the above cases, 
			an adversary can deduce that at least one of $v_0$, $v_1$, and $v_2$ must be an accepting state. 
			Therefore,
			$\Obj^{1,\langle v_0 \rangle}_{\Buchi,\PW}(\sigma_3)
			  =\{\langle v_0 \rangle,\langle v_1 \rangle, \langle v_2 \rangle,
			   \langle v_0,v_1\rangle,\langle v_1,v_2 \rangle, \langle v_2,v_0 \rangle,
			   \langle v_0,v_1,v_2 \rangle\}$.
			For all $\calO\subseteq\Obj^{1,\langle v_0 \rangle}_{\Buchi,\PW}(\sigma_3)$,
			$\sigma_3$ is an $\calO$-IS. 
	\end{itemize}
	In the above example, 
	$\Obj^{1,\langle v_0 \rangle}_{\Buchi,\PW}(\sigma_1) 
		\subset \Obj^{1,\langle v_0 \rangle}_{\Buchi,\PW}(\sigma_2)
		\subset \Obj^{1,\langle v_0 \rangle}_{\Buchi,\PW}(\sigma_3)$.
	Hence, the strategy $\sigma_3$ is the most favorable one for player $1$
	with regard to her privacy protection.
	This observation motivates us to introduce a new concept of equilibrium defined below.
\end{example}

	\paragraph*{Objective-indistinguishability equilibrium}
	\begin{definition}\label{dfn:OIE}
	Let $(O_p)_{p\in P}$ be an objective profile and 
	$\Knw \in \{\PW, \GW, \PG, \PGW\}$ be a type of information that an adversary can use.
We call a strategy profile $\bmsigma \in \Sigma$ such that
\begin{equation}
	\forall p\in P. \ \forall\sigma_p \in \Sigma^p. \ \Objpk(\bmsigma[p\mapsto\sigma_p])\subseteq\Objpk(\bmsigma)
	\label{eqn:OIE}
\end{equation}
an \emph{objective-indistinguishability equilibrium} (OIE) for $\Knw$.
\end{definition}
If a strategy profile $\bmsigma$ is an OIE for $\Knw$, no player can expand her $\Objpk(\bmsigma)$
by changing her strategy alone.
For a strategy profile $\bmsigma\in\Sigma$, 
we call a strategy $\sigma_p\in\Sigma^p$ such that 
$\Objpok(\bmsigma[p\mapsto\sigma_p])\not\subseteq\Objpok(\bmsigma)$ a profitable deviation for OIE.
If an OIE $\bmsigma$ is an NE as well, we call $\bmsigma$ an \emph{objective-indistinguishability Nash equilibrium} (OINE).
\begin{example}
	Figure \ref{fig:OIE} shows a $3$-player game arena $\G=(P,V,(V_p)_{p\in P},v_0,E)$
	where $P=\{0,1,2\}$, $V=\{v_0,v_1,v_2\}$, $V_p=\{v_p\} \ (p\in P)$ and
	$E= \{ (v_i,v_j) \mid i,j\in P,i\neq j \}$.
	\begin{figure}[t]
		\centering
		\begin{tikzpicture}[everynode/.style={circle,draw,minimum size=0.8cm},>=stealth,node distance=2.5cm]
			\node[everynode] (v0) {$v_0$};
			\node[left of=v0,node distance=1.2cm] (start) {};
			\node[everynode,below left of=v0] (v1) {$v_1$};
			\node[everynode,below right of=v0] (v2) {$v_2$};
			\draw[->] (v0) to [bend right=15] (v1);
			\draw[->] (v1) to [bend right=15] (v2);
			\draw[->] (v2) to [bend right=15] (v0);
			\draw[->] (v0) to [bend right=15] (v2);
			\draw[->] (v1) to [bend right=15] (v0);
			\draw[->] (v2) to [bend right=15] (v1);
			\draw[->] (start) to (v0);
		\end{tikzpicture}
			\caption{$3$-player game arena with B\"{u}chi objectives}
			\label{fig:OIE}
	\end{figure}
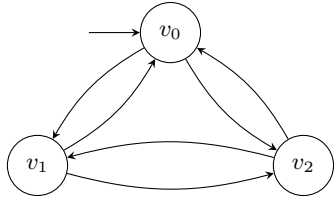
	The objective of player $p\in P$ is $\langle v_p \rangle$, and hence
	the objective profile is $\bmalpha=(\langle v_0 \rangle,\langle v_1 \rangle, \langle v_2 \rangle)$.
	Let $\sigma_p\in\Sigma^p\ (p\in P)$ be the strategies defined as follows:
	$\sigma_0(hv_0) = v_1$, $\sigma_1(hv_1) = v_0$, and
	$\sigma_2(hv_2) = v_0$ for every $h\in \Hist \cup \{\varepsilon\}$.
	Let $\bmsigma= (\sigma_1,\sigma_2,\sigma_3)$.
	It holds that $\out(\bmsigma)=v_0v_1v_0v_1\cdots$ and $\Win(\bmsigma,\bmalpha)=\{0,1\}$.
	\begin{itemize}
		\item For $\Knw=\PW$, $\bmsigma$ is not an OIE 
			  because there exists a profitable deviation $\sigma'_1\in\Sigma^1$ for OIE
			  such that $\sigma'_1(h)=v_2$ for all $h\in\Hist_1$.
			  While $\out(\bmsigma)$ does not visit $v_2$, 
			  player $1$ can make the outcome visit $v_2$ infinitely often 
			  by changing her strategy from $\sigma_1$ to $\sigma'_1$. 
			  As a result,
			  $\Obj^{1,\langle v_1 \rangle}_{\Buchi,\PW}(\bmsigma)=
			  \{ \langle v_0 \rangle, \langle v_1 \rangle, \langle v_0,v_1 \rangle, \langle v_1,v_2 \rangle,
			  \langle v_2,v_0 \rangle, \langle v_0,v_1,v_2 \rangle\}$ and
			  $\Obj^{1,\langle v_1 \rangle}_{\Buchi,\PW}(\bmsigma[1\mapsto\sigma'_1])
			  =\Obj^{1,\langle v_1 \rangle}_{\Buchi,\PW}(\bmsigma) \cup \{ \langle v_2 \rangle\}$.
		\item 
			  For $\Knw = \GW$, $\bmsigma$ is an OIE by the following reason: 
			  In general, when $\Knw = \GW$, 
			  by definition 
			  $\Obj^{p,O_p}_{\Omega,\GW}(\bmsigma) = \Omega\setminus \{\varnothing\}$ if $p$ wins 
			  and $\Obj^{p,O_p}_{\Omega,\GW}(\bmsigma) = \overline{\Winnable^p}$ otherwise. 
			  (That is, an adversary cannot exclude any objective other than $\varnothing$ 
			  from candidate objectives of player $p$ when $p$ wins, 
			  while he can exclude objectives in $\Winnable^p$ when $p$ loses.) 
			  In this example, 
			  $\Obj^{0,\langle v_0 \rangle}_{\Omega,\GW}(\bmsigma) 
				= \Obj^{1,\langle v_1 \rangle}_{\Omega,\GW}(\bmsigma) = \Omega \setminus \{\varnothing\}$ 
			  since players $0$ and $1$ are winners. 
			  They have no profitable deviation for OIE,
			  because each of them cannot become a loser unless other players change their strategies
			  and thus 
			  $\Obj^{p,\langle v_p \rangle}_{\Omega,\GW}(\bmsigma[p\mapsto\sigma'_p]) \ (p\in\{0,1\})$ 
			  still equals $\Omega \setminus \{\varnothing\}$ 
			  for any strategy $\sigma'_p \ (p\in\{0,1\})$.
			  For player~$2$, 
			  $\Obj^{2,\langle v_2 \rangle}_{\Omega,\GW}(\bmsigma) 
			    = \overline{\Winnable^2}$ ($= \{ \langle \rangle, \langle v_2 \rangle \})$.\footnote{%
				  In this example, player $2$ can visit $v_i \ (i=0,1)$ infinitely often
				  by choosing $v_i$ as the next vertex at $v_2$.
				  Therefore, an objective such that $v_0$ or $v_1$ is an accepting state is winnable
				  and hence $\Winnable^2 = \Omega \setminus \{ \langle \rangle, \langle v_2 \rangle \}$.
			  }
			  She also has no profitable deviation for OIE, 
			  because she cannot become a winner unless player $0$ or $1$ changes their strategies 
			  and thus $\Obj^{2,\langle v_2 \rangle}_{\Omega,\GW}(\bmsigma[2\mapsto \sigma'_2])$ still equals 
			  $\overline{\Winnable^2}$ for any her strategy $\sigma'_2$.
		\item For $\Knw=\PG$, $\bmsigma$ is not an OIE because for $\sigma'_1\in\Sigma^1$ defined above,
			  $\Obj^{1,\langle v_1 \rangle}_{\Buchi,\PG}(\bmsigma)=
			  \{ \langle\rangle,\langle v_0 \rangle, \langle v_1 \rangle, \langle v_0,v_1 \rangle, \langle v_1,v_2 \rangle,
			  \langle v_2,v_0 \rangle, \langle v_0,v_1,v_2 \rangle\}$ and 
			  $\Obj^{1,\langle v_1 \rangle}_{\Buchi,\PG}(\bmsigma[1\mapsto\sigma'_1])
			  =\Obj^{1,\langle v_1 \rangle}_{\Buchi,\PG}(\bmsigma) \cup \{ \langle v_2 \rangle\}$.
		\item For $\Knw=\PGW$, $\bmsigma$ is not an OIE because 
			  $\sigma'_1\in\Sigma^1$ is again a profitable deviation for OIE.
	\end{itemize}
\end{example}

\section{Decidability Results}\label{sec:results}
	\begin{theorem}
  \label{th:existence-of-calO-IS}
  Let $\calG = (P, V, (V_p)_{p\in P}, v_0, E)$ be a game arena
  and $\balpha = (O_p)_{p\in P}$ be an objective profile
  over Muller objectives.
  For a subset $\calO \subseteq \powerset{\Play}$ of
  Muller objectives,
  whether there exists an $\calO$-IS of $p$
  for $O_p$ is decidable.
  Moreover, the problem is decidable in polynomial time
  when $\Knw=\PG$ or when $\Knw=\GW$ and $\calO$ does not contain $\varnothing$.
\end{theorem}
\begin{proof}
  First we consider the case where $\Knw = \PGW$.
  We can show that
  a strategy $\sigma_p\in \Sigma^p$ is an $\calO$-IS
  of $p$ for $O_p$,
  i.e.\ $\calO \subseteq
  \bigcap_{\bsigma\in\Sigma}
    \Obj^{p,O_p}_{\PGW}(\bsigma[p\mapsto\sigma_p])$,
  if and only if
  \begin{equation}
    \Out^p(\sigma_p)\subseteq
    \bigcap_{O\in\calO}
      \left((O \cap O_p) \cup (\overline{O} \cap \overline{O_p})\right)
    \cap \bigcap_{O\in\calO\cap\Winnable^p} (O \cap O_p).
    \label{eqn:calO-IS}
  \end{equation}
  This can be shown as follows:\footnote{%
    We have confirmed this equivalence using
    a proof assistant software Coq.
    The proof script is available at
    \url{https://github.com/ytakata69/proof-indistinguishable-objectives}.
  }
  Assume that
  $\calO \subseteq
  \bigcap_{\bsigma\in\Sigma}
    \Obj^{p,O_p}_{\PGW}(\bsigma[p\mapsto\sigma_p])$.
  Then, every $O\in\calO$ should belong to
  $\Obj^{p,O_p}_{\PGW}(\bsigma[p\mapsto\sigma_p])$
  for every $\bsigma\in\Sigma$.
  Then by the definition of $\Obj^{p,O_p}_{\PGW}$,
  every $O\in\calO$ and every $\bsigma\in\Sigma$
  should satisfy
  $\Out(\bsigma[p\mapsto\sigma_p])\in
    (O\cap O_p)\cup(\overline{O}\cap\overline{O_p})$ and
  whenever $O\in\Winnable^p$,
  $\Out(\bsigma[p\mapsto\sigma_p])\in
    O\cap O_p$.
  Because 
  $\Out^p(\sigma_p) = \{ \Out(\bsigma[p\mapsto\sigma_p]) \mid
    \bsigma\in\Sigma \}$,
  we have Condition~(\ref{eqn:calO-IS}).
  The reverse direction can be proved similarly.

  Condition~(\ref{eqn:calO-IS}) means that
  $\sigma_p$ is a winning strategy of $p$
  for the objective equal to the right-hand side
  of the containment in Condition~(\ref{eqn:calO-IS}).
  Because the class of Muller objectives is closed under
  Boolean operations,
  the right-hand side of Condition~(\ref{eqn:calO-IS}) is
  also a Muller objective.
  Since 
  deciding the existence of a winning strategy
  for a Muller objective is decidable as stated in Theorem \ref{thm:Muller},
  deciding the existence of an $\calO$-IS is also
  decidable.
  (In this computation,
  deciding the existence of a winning strategy is
  used both for deciding whether $O\in\Winnable^p$, i.e.,\
  $O$ has a winning strategy, and for deciding whether
  the right-hand side of Condition~(\ref{eqn:calO-IS}) has
  a winning strategy.)

  For the other cases, 
  we can similarly show that $\sigma_p\in \Sigma^p$ is
  an $\calO$-IS of $p$ for $O_p$ if and only if
  the following conditions
  (\ref{eqn:calO-IS-PW}),
  (\ref{eqn:calO-IS-GW}), and
  (\ref{eqn:calO-IS-PG}) hold
  when $\Knw = \PW,\GW,\PG$, respectively:
  \begin{align}
    \Out^p(\sigma^p) &\subseteq
    \bigcap_{O\in\calO}
      \left((O \cap O_p) \cup (\overline{O} \cap \overline{O_p})\right),
    \label{eqn:calO-IS-PW}
    \\
    \Out^p(\sigma^p) &\subseteq
    \bigcap_{O\in\calO\cap\Winnable^p} O_p \cap
	  \bigcap_{O\in\calO\cap\{\varnothing\}} \overline{O_p},
    \label{eqn:calO-IS-GW}
    \\
    \Out^p(\sigma^p) &\subseteq
    \bigcap_{O\in\calO\cap\Winnable^p} O.
    \label{eqn:calO-IS-PG}
  \end{align}
  Therefore in any cases,
  we can reduce the problem of deciding
  the existence of an $\calO$-IS into
  the one deciding the existence of
  a winning strategy for a Muller objective.

  Since
  $\Muller(\calF_1)\cap\Muller(\calF_2)=\Muller(\calF_1\cap\calF_2)$,
  the description lengths of the right-hand sides of Condition
  (\ref{eqn:calO-IS-PG}) and 
  Condition (\ref{eqn:calO-IS-GW}) with $\calO$ not containing $\varnothing$
  are not greater than the sum of those of $\calO$ and $O_p$.\footnote{%
    As an exception, if $\calO\cap\Winnable^p=\emptyset$
	(resp. $\calO \cap (\Winnable^p\cup \{ \varnothing \}) = \varnothing$),
    then the right-hand side of
	Condition (\ref{eqn:calO-IS-PG}) (resp. (\ref{eqn:calO-IS-GW}))
    equals the set of all plays,
    which equals $\Muller(2^V)$.
    In these cases,
    every strategy satisfies
    Conditions (\ref{eqn:calO-IS-GW}) and (\ref{eqn:calO-IS-PG})
    and thus we can trivially decide the existence of
    an $\calO$-IS\@.
  }
  Since
  deciding the existence of a winning strategy
  for a Muller objective is solvable in polynomial time by Theorem \ref{thm:Muller},
  deciding the existence of an $\calO$-IS
  when $\Knw=\PG$ or when $\Knw=\GW$ and $\calO$ does not contain $\varnothing$ is also solvable in polynomial time.
\qed
\end{proof}

When $\Knw=\PGW$ or $\PW$,
we cannot guarantee that
deciding the existence of an $\calO$-IS is
solvable in polynomial time
because
the complementation of a Muller objective
in the right-hand sides of
Conditions (\ref{eqn:calO-IS}) and (\ref{eqn:calO-IS-PW})
may make the description length of the resultant objective
$O(|V|\cdot 2^{|V|})$
even when the description lengths of
$\calO$ and $\balpha$ are small.
Similarly, when $\Knw=\GW$, $\calO\cap \Winnable^p=\varnothing$ and $\varnothing\in\calO$,
we cannot guarantee that
deciding the existence of an $\calO$-IS is solvable in polynomial time
because
the right-hand side of Condition (\ref{eqn:calO-IS-GW}) becomes $\overline{O_p}$.

\begin{theorem}
  Let $\calG = (P, V, (V_p)_{p\in P}, v_0, E)$ be a game arena
  and $\balpha = (O_p)_{p\in P}$ be an objective profile
  over Muller objectives.
  For a subset $\calO \subseteq \powerset{\Play}$ of
  Muller objectives,
  whether there exists a winning $\calO$-IS of $p$
  for $O_p$ is decidable in polynomial time.
\end{theorem}
\begin{proof}
  By definition, $\sigma_p\in\Sigma^p$ is a winning strategy of $p$
  for $O_p$ if and only if $\Out^p(\sigma_p)\subseteq O_p$.
  Therefore, by replacing the right-hand side of each of
  Conditions (\ref{eqn:calO-IS})--(\ref{eqn:calO-IS-PG}) with
  the intersection of it and $O_p$,
  we can decide the existence of a winning $\calO$-IS
  in the same way as the proof of Theorem~\ref{th:existence-of-calO-IS}.
  Namely, $\sigma_p$ is a winning $\calO$-IS of $p$ for $O_p$
  if and only if
  \begin{align}
    \Out^p(\sigma^p) &\subseteq
    O_p \cap \bigcap_{O\in\calO} O
	  && \text{(when $\Knw=\PGW$ or $\PW$)}, \label{eqn:WO-IS-proof-PGWorPW}
    \\
    \Out^p(\sigma^p) &\subseteq
	  O_p \cap \bigcap_{O\in \calO\cap\{\varnothing\}} \overline{O_p} && \text{(when $\Knw=\GW$)},\label{eqn:WO-IS-proof-GW}
    \\
    \Out^p(\sigma^p) &\subseteq
    O_p \cap \bigcap_{O\in\calO\cap\Winnable^p} O
	  && \text{(when $\Knw=\PG$)}.\label{eqn:WO-IS-proof-PG}
  \end{align}

  When $\Knw=\PGW,\PW$ or $\PG$,
  since the right-hand sides of Conditions (\ref{eqn:WO-IS-proof-PGWorPW}) and (\ref{eqn:WO-IS-proof-PG})
  do not require complementation,
  the description lengths of them are not greater than
  the sum of the description lengths of $\calO$ and $O_p$.
	When $\Knw=\GW$, the right-hand side of Condition (\ref{eqn:WO-IS-proof-GW})
	is $O_p \cap \overline{O_p}=\varnothing$ if $\varnothing\in\calO$, and
	$O_p$ otherwise,
	and hence the description length of it is not greater than the description length of $O_p$.
  Therefore,
  in the same way as the cases where $\Knw=\PG$ or $\Knw=\GW$ and $\varnothing\notin\calO$
  in Theorem~\ref{th:existence-of-calO-IS},
  deciding the existence of a winning $\calO$-IS
  is also solvable in polynomial time
  for any $\Knw\in\{\PW,\GW,\PG,\PGW\}$.
\qed
\end{proof}

	\begin{theorem}\label{thm:OIE}
	For a game arena $\G$ and an objective profile $\bmalpha=(O_p)_{p\in P}$ over Muller objectives, 
	whether there exists an OIE for $\G$ and $\bmalpha$ is decidable.
\end{theorem}

\begin{proof}
Condition (\ref{eqn:OIE}) in Definition \ref{dfn:OIE} is equivalent to the following condition:
\begin{equation}
	\forall p\in P.\  \forall \sigma_p \in \Sigma^p.\  \forall O \in \Omega.\
	O \in \Objpok(\bmsigma[p\mapsto \sigma_p]) \Rightarrow O\in\Objpok(\bmsigma).
	\label{eqn:OIE-proof}
\end{equation}
First we consider the case where $\Knw=\PGW$.
By the definition of $\Objp_{\Omega,\PGW}$,
Condition (\ref{eqn:OIE-proof}) for $\Knw=\PGW$ is equivalent to the following condition:
\begin{equation}
	\begin{aligned}
	&\forall p\in P.\  \forall \sigma_p \in \Sigma^p.\  \forall O \in \Omega.\ \\
		&\text{if $O\in\Winnable^p$, then $(\out(\bmsigma[p\mapsto\sigma_p])\in O\cap O_p\Rightarrow \out(\bmsigma)\in O\cap O_p);$} \\
		&\text{otherwise, $(\out(\bmsigma[p\mapsto\sigma_p])\in(O\cap O_p)\cup(\overline{O}\cap\overline{O_p})
		\Rightarrow \out(\bmsigma) \in (O\cap O_p)\cup(\overline{O}\cap \overline{O_p}))$}.
	\end{aligned}\label{eqn:OIE-proof2}
\end{equation}
For $O\in\calO$ and $p\in P$, let $R^O_p$ be the  objective defined as follows:
\begin{equation*}
	R_p^O=\begin{cases}
		O\cap O_p & O\in\Winnable^p, \\
		(O\cap O_p) \cup (\overline{O}\cap\overline{O_p}) & O\notin\Winnable^p.
			\end{cases}
\end{equation*}
Let $\bmalpha_O = (R^O_p)_{p\in P}$ be the objective profile consisting of these objectives.
Then, Condition (\ref{eqn:OIE-proof2}) can be written as $\forall O\in\calO.\ \Nash(\bmsigma,\bmalpha_O)$.
Therefore, this theorem holds by Theorem \ref{thm:n-nash}.

For the other cases, the implication inside the scope of the three universal quantifiers 
in Condition (\ref{eqn:OIE-proof}) is equivalent to the following implications:
\begin{equation*}
	\begin{aligned}
	&\text{when $\Knw=\PW$}\\
	&\out(\bmsigma[p\mapsto\sigma_p])\in (O\cap O_p)\cup(\overline{O}\cap\overline{O_p})
		\Rightarrow\out(\bmsigma)\in (O\cap O_p)\cup(\overline{O}\cap\overline{O_p}), \\
	&\text{when $\Knw=\GW$} \\
	&\text{if $O\in\Winnable^p$, then }
		\out(\bmsigma[p\mapsto\sigma_p])\in O_p\Rightarrow\out(\bmsigma)\in O_p; \\
	&\text{if $O=\varnothing$, then }
		\out(\bmsigma[p\mapsto\sigma_p])\in \overline{O_p}\Rightarrow\out(\bmsigma)\in \overline{O_p}, \\
	&\text{when $\Knw=\PG$}\\
	&\text{if $O\in\Winnable^p$, then }
		\out(\bmsigma[p\mapsto\sigma_p])\in O\Rightarrow\out(\bmsigma)\in O. \\
	\end{aligned}
\end{equation*}
These conditions can be written as the combination of NE in the same way as the case where $\Knw = \PGW$.
Therefore, this theorem also holds for $\Knw\in\{\PW,\GW,\PG\}$ by Theorem \ref{thm:n-nash}.
\qed
\end{proof}

\begin{theorem}\label{thm:OINE}
	For a game arena $\G$ and an objective profile $\bmalpha=(O_p)_{p\in P}$ over Muller objectives, 
	whether there exists an OINE for $\G$ and $\bmalpha$ is decidable.
\end{theorem}
\begin{proof}
	By the proof of Theorem \ref{thm:OIE}, an OINE $\bmsigma\in\Sigma$ must satisfy the condition
	$\forall O\in\calO.\ \Nash(\bmsigma,\bmalpha_O)$.
	Moreover, $\bmsigma$ must also satisfy $\Nash(\bmsigma,\bmalpha)$
	because $\bmsigma$ is a NE.
	Therefore, $\bmsigma$ is an $((\bmalpha_O)_{O\in\calO},\bmalpha)$-NE and thus,
	this theorem holds by Theorem \ref{thm:n-nash}.
	\qed
\end{proof}

\section{Conclusion}\label{sec:conclusion}
	We proposed two new notions 
$\calO$-indistinguishable strategy ($\calO$-IS) and objective-indistinguishability equilibrium (OIE).
Then, we proved that whether there exists an $\calO$-IS and an OIE over Muller objectives are both decidable.
To prove this, we defined an $(\bmalpha_1,\ldots,\bmalpha_n)$-Nash equilibrium
as a strategy profile which is simultaneously a nash equilibrium for all objective profiles $\bmalpha_1,\ldots,\bmalpha_n$,
and proved that whether there exists an $(\bmalpha_1,\ldots,\bmalpha_n)$-Nash equilibrium is decidable. 

In this paper, we assume that an adversary is not a player 
but an individual who observes partial information on the game.
He cannot directly affect the outcome of the game by choosing next vertices.
We can consider another setting where an adversary is also a player.
His objective is minimizing the set $\Objpok$ of candidate objectives of other players and
he takes a strategy for achieving the objective.
Considering a framework on this setting, by extending the results shown in this paper, is future work.

\bibliographystyle{plain}
\bibliography{refs}

\begin{thebibliography}{10}

\bibitem{ACGMMTZ16}
M.~Abadi, A.~Chu, I.~Goodfellow, H.~B. McMahan, I.~Moronov, K.~Talwar, and
  L.~Zhang.
\newblock Deep learning with differential privacy.
\newblock {\em ACM CCS}, 2016.

\bibitem{AG22}
S.~Almagor and S.~Guendelman.
\newblock Concurrent games with multiple topologies.
\newblock arXiv: 2207.02596.

\bibitem{ABCP13}
M.~E. Andr\'{e}s, N.~E. Bordenabe, K.~Chatzikokolakis, and C.~Palamidessi.
\newblock Geo-indistinguishability: Differential privacy for location based
  systems.
\newblock {\em ACM CCS}, 2013.

\bibitem{BA05}
R.~J. Bayardo and R.~Agrawal.
\newblock Data privacy through optimal $k$-anonymization.
\newblock {\em ICDE}, pages 217--228, 2005.

\bibitem{BMMRY21}
R.~Berthon, B.~Maubert, A.~Murano, S.~Rubin, and M.~Y. Vardi.
\newblock Strategy logic with imperfect information.
\newblock {\em ACM Trans. Computational Logic}, 22(1):1--51, 2021.

\bibitem{BCJ18}
R.~Bloem, K.~Chatterjee, and B.~Jobstmann.
\newblock Graph games and reactive synthesis.
\newblock In E.~M.~Clarke et~al., editor, {\em Handbook of Model Checking},
  chapter~27, pages 921--962. Springer, 2018.

\bibitem{BMV17}
P.~Bouyer, N.~Markey, and S.~Vester.
\newblock Nash equilibria in symmetric graph games with partial observation.
\newblock {\em Information and Computation}, 254:238--258, 2017.

\bibitem{Br17}
V.~Bru\'{e}re.
\newblock Computer aided synthesis: a game-theoretic approach.
\newblock {\em DLT}, pages 3--35, 2017.

\bibitem{BKBL07}
J.-W. Byun, A.~Kamra, E.~Bertino, and N.~Li.
\newblock Efficient $k$-anonymization using clustering techniques.
\newblock {\em DASFAA}, pages 188--200, 2007.

\bibitem{CAH05}
K.~Chatterjee, L.~de~Alfaro, and T.~A. Henzinger.
\newblock The complexity of stochastic rabin and streett games.
\newblock {\em ICALP}, 2005.

\bibitem{CD10}
K.~Chatterjee and L.~Doyen.
\newblock The complexity of partial-observation parity games.
\newblock {\em LPAR}, pages 1--14, 2010.

\bibitem{CHP14}
K.~Chatterjee and L.~Doyen.
\newblock Games with a weak adversary.
\newblock {\em ICALP}, pages 110--121, 2014.

\bibitem{CDFR17}
K.~Chatterjee, L.~Doyen, E.~Filiot, and J.-F. Raskin.
\newblock Doomsday equilibria for omega-regular games.
\newblock {\em Information and Computation}, 254:296--315, 2017.

\bibitem{CHJ06}
K.~Chatterjee, T.~A. Henzinger, and M.~Jurdzi\'{n}ski.
\newblock Games with secure equilibria.
\newblock {\em Theoretical Computer Science}, 365:67--82, 2006.

\bibitem{CPP08}
K.~Chatzikokolakis, C.~Palamidessi, and P.~Panangaden.
\newblock Anonymity protocols as noisy channels.
\newblock {\em Information and Computation}, 206(2-4):378--401, 2008.

\bibitem{CHM07}
D.~Clark, S.~Hunt, and P.~Malacaria.
\newblock A static analysis for quantifying information flow in a simple
  imperative language.
\newblock {\em J. Computer Security}, 15:321--371, 2007.

\bibitem{Dw06}
C.~Dwork.
\newblock Differential privacy.
\newblock {\em ICALP}, pages 1--12, 2006.

\bibitem{Dw08}
C.~Dwork.
\newblock Differential privacy: A survey of results.
\newblock {\em TAMC}, pages 1--19, 2008.

\bibitem{DMNS06}
C.~Dwork, F.~D. McSherry, K.~Nissim, and A.~Smith.
\newblock Calibrating noise to sensitivity in private data analysis.
\newblock {\em TCC}, pages 265--284, 2006.

\bibitem{DR14}
C.~Dwork and A.~Roth.
\newblock The algorithmic foundations of differential privacy.
\newblock {\em Foundations and Trends in Theoretical Computer Science}, 9:3--4,
  2013.
\newblock now Publishers.

\bibitem{FKL10}
D.~Fisman, O.~Kupferman, and Y.~Lustig.
\newblock Rational synthesis.
\newblock {\em TACAS}, pages 190--204, 2010.

\bibitem{FWCY10}
B.~C.~M. Fung, K.~Wang, R.~Chen, and P.~S. Yu.
\newblock Privacy-preserving data publishing: A survey of recent developments.
\newblock {\em ACM Computing Surveys}, 42(4):14:1--14:53, June 2010.

\bibitem{Go01}
O.~Goldreich.
\newblock {\em Foundations of Cryptography}, volume I Basic Tools.
\newblock Cambridge University Press, 2001.

\bibitem{KL22}
O.~Kupferman and O.~Leshkowitz.
\newblock Synthesis of privacy-preserving systems.
\newblock {\em FSTCS}, 42:1--21, 2022.

\bibitem{LLV07}
N.~Li, T.~Li, and S.~Venkatasubramanian.
\newblock $t$-closeness: Privacy beyond $k$-anonymity and ${\ell}$-diversity.
\newblock {\em ICDE}, pages 106--115, 2007.

\bibitem{MGK07}
A.~Machanavajjhala, J.~Gehrke, and D.~Kifer.
\newblock ${\ell}$-diversity: Privacy beyond $k$-anonymity.
\newblock {\em ICDE}, 24, 2006.
\newblock also in TKDD, 1(1), Mar 2007.

\bibitem{PR89}
A.~Pnueli and R.~Rosner.
\newblock On the synthesis of a reactive module.
\newblock {\em ACM POPL}, pages 179--190, 1989.

\bibitem{Sa01}
P.~Samarati.
\newblock Protecting respondents' identities in microdata release.
\newblock {\em IEEE Trans. Knowledge and Data Engineering}, 13(6):1010--1027,
  2001.

\bibitem{SS15}
R.~Shokri and V.~Shmatikov.
\newblock Privacy-preserving deep learning.
\newblock {\em ACM CCS}, 2015.

\bibitem{SSSS17}
R.~Shokri, M.~Stronati, C.~Song, and V.~Shmatikov.
\newblock Membership inference attacks against machine learning models.
\newblock {\em IEEE Symp. Security and Privacy}, 2017.

\bibitem{Sm09}
G.~Smith.
\newblock On the foundations of quantitative information flow.
\newblock {\em FoSSaCS}, pages 288--302, 2009.

\bibitem{Sw02}
L.~Sweeney.
\newblock $k$-anonymity: A model for protecting privacy.
\newblock {\em Int'l Journal on Uncertainty, Fuzziness and Knowledge-based
  Systems}, 10(5):557--570, 2002.

\bibitem{Um08}
M.~Ummels.
\newblock The complexity of nash equilibria in infinite multiplayer games.
\newblock {\em FOSSACS}, pages 20--34, 2008.

\bibitem{UW11}
M.~Ummels and D.~Wojtczak.
\newblock The complexity of nash equilibria in stochastic multiplayer games.
\newblock {\em Logical Methods in Computer Science}, 7(3), 2011.

\end{thebibliography}

\appendix
\newpage
\section*{Appnedix}
	An objective $O\subseteq 2^\Play$ is \emph{prefix-independent}
if $\rho \in O \Leftrightarrow h\rho \in O$ for every play $\rho\in O$ and history $h \in \Hist$.
The objectives defined in Definition \ref{def:obj} are prefix-independent 
because $\Inf(\rho) = \Inf(h\rho)$ for every play $\rho$ and history $h$.
For a game arena $\G=\Arena$ and $v\in V$,
let $(\G,v)= (P,V,(V_p)_{p\in P},v,E)$ be the game arena obtained from $\G$ by replacing
the initial vertex $v_0$ of $\G$ with $v$.

For a game arena $\G=\Arena$ with an objective profile $\bmalpha=(O_p)_{p\in P}$,
we define the game arena $\G_p=(\{p,-p\},V,(V_p,\overline{V_p}),v_0,E)$ with
the objective profile $(O_p,\overline{O_p})$ for each $p\in P$.
The game arena $\G_p$ with the objective profile $(O_p,\overline{O_p})$
is a $2$-player zero-sum game such that 
vertices and edges are the same as $\G$ and
the player $-p$ is formed by the \emph{coalition} of all the players in $P \setminus \{p\}$.
The following proposition is a variant of \cite[Proposition 28]{Br17} adjusted to the settings of this paper.
\begin{proposition}\label{prop:NE}
	Let $\G=\Arena$ be a game arena and $\bmalpha=(O_p)_{p\in P}$ be an objective profile such that
	$O_p$ is prefix-independent for all $p$.
	Then, a play $\rho=v_0v_1v_2\cdots\in\Play$ is the outcome of some NE $\bmsigma\in\Sigma$ for $\bmalpha$,
	i.e., $\rho=\out(\bmsigma)$,
	if and only if 
	$\forall p \in P.\ \forall i\geq0.\ 
		(v_i\in V_p \wedge O_p \in \Winnable^p_{(\G,v_i)}) \Rightarrow v_iv_{i+1}v_{i+2}\cdots \in O_p$.
\end{proposition}
\begin{proof}
($\Rightarrow$) We prove this direction by contradiction.
	Assume that a play $\rho=v_0v_1v_2\cdots\in\Play$ is the outcome of 
	a NE $\bmsigma=(\sigma_p)_{p\in P}\in\Sigma$ for $\bmalpha$
	and there exist $p\in p$ and $i\geq 0$ with $v_i \in V_p \wedge O_p \in \Winnable^p_{(\G,v_i)}$ 
	such that $v_i v_{i+1} v_{i+2} \cdots \notin O_p$.
	By the prefix-independence of $O_p$, $\rho=\out(\bmsigma)=v_0v_1v_2\cdots\notin O_p$ 
	and thus $p\notin\Win_{\G}(\bmsigma,\bmalpha)$.
	Since $O_p \in \Winnable^p_{(\calG,v_i)}$, there exists a winning strategy $\tau_p$ of $p\in (\calG,v_i)$.  
	Let $\sigma'_p$ be the strategy obtained from $\sigma_p$ and $\tau_p$ as follows: 
	Until producing $v_0v_1\cdots v_i$, $\sigma'_p$ is the same as $\sigma_p$.  
	From $v_i$, $\sigma'_p$ behaves as the same as $\tau_p$. 
	Therefore, $\out(\bmsigma[p\mapsto\sigma'_p])$ equals $v_0v_1\cdots v_{i-1}\pi$ for some play $\pi$ of $(\calG,v_i)$, 
	and $\pi \in O_p$ because $\tau_p$ is a winning strategy of $p$ in $(\calG,v_i)$.
	From prefix-independence of $O_p$ it follows that $\out(\bmsigma[p\mapsto\sigma'_p])\in O_p$.
	This contradicts the assumption that $\bmsigma$ is an NE\@.

($\Leftarrow$) Let $\rho=v_0v_1v_2\cdots\in\Play$ be a play on $\G$ 
	and assume that $v_iv_{i+1}v_{i+2}\cdots \in O_p$
	for all $p \in P$ and $i\geq 0$ such that $v_i\in V_p \wedge O_p \in \Winnable^p_{(\G,v_i)}$.
	We define a strategy profile $\bmsigma=(\sigma_p)_{p\in P}$ 
	as the one that satisfies the following two conditions:
	First, $\bmsigma$ produces $\rho$ as its outcome,
	i.e., $\out(\bmsigma)=\rho$.
	Second, if some player $p$ deviates from $\rho$ at $v_j \in V_p \ (j\geq 0)$
	and $O_p\notin\Winnable^p_{(\G,v_j)}$,
	then all the other players (as a coalition) play from $v_j$ 
	according to a winning strategy of $-p$ for $(\G_p,v_j)$ and $\overline{O_p}$.
	(Note that in a $2$-player zero-sum game, there is always a winning strategy for one of the players,
	and thus there is a winning strategy of $-p$ for $(\G_p,v_j)$ and $\overline{O_p}$ when $O_p\notin\Winnable^p_{(\G,v_j)}$.)
	We can show that the strategy profile $\bmsigma$ is a NE as follows:
	Assume that some player $p$ deviates from $\sigma_p$ to a strategy $\sigma'_p\in\Sigma^p$,
	and $\out_\G(\bmsigma[p\mapsto\sigma'_p])$ deviates from $\rho$ at $v_j\in V_p$ for some $j\geq0$.
	If $O_p\in\Winnable^p_{(\calG,v_j)}$, then by assumption, $v_jv_{j+1}v_{j+2}\cdots\in O_p$. 
	By the prefix-independence, $\rho=\out_{\calG}(\bmsigma)\in O_p$ and thus 
	$\sigma'_p$ is not a profitable deviation. 
	Otherwise, as described above,
	all the other players (as a coalition) punish the player $p$ by taking a winning strategy
	of $-p$ for $(\G_p,v_j)$ and $\overline{O_p}$, 
	and hence $p\notin\Win(\bmsigma[p\mapsto\sigma'_p])$.
	Therefore $\sigma'_p$ is not a profitable deviation also in this case.
	\qed
\end{proof}

\begin{corollary}\label{cor:n-nash}
	Let $\G=\Arena$ be a game arena and $\bmalpha_j=(O_p^j)_{p\in P}$ $(1\leq j\leq n)$ be objective profiles
	such that $O_p^j\subseteq\Play$ is prefix-independent for all $p\in P$ and $1\leq j\leq n$.
	Then, a play $\rho=v_0v_1v_2\cdots \in\Play$ is the outcome of 
	some $(\bmalpha_1,\ldots,\bmalpha_n)$-NE $\bmsigma\in\Sigma$, i.e., $\rho=\out(\bmsigma)$,
	if and only if
	\begin{equation}\label{con:n-nash}
		\begin{aligned}
			&\forall p\in P.\ \forall i\geq 0.\ 1\leq \forall j\leq n.\ \\
			&(v_i\in V_p \wedge O_p^j\in\Winnable^p_{(\G,v_i)}) \Rightarrow v_iv_{i+1}v_{i+2}\cdots \in O_p^j.
		\end{aligned}
	\end{equation}
\end{corollary}
Corollary \ref{cor:n-nash} can be easily proved by Proposition \ref{prop:NE} and Definition \ref{def:n-nash}.
\par\medskip
\noindent\textbf{Theorem \ref{thm:n-nash}. }\textit{
Let $\G=\Arena$ be a game arena and $\bmalpha_j=(O_p^j)_{p\in P}$ $(1\leq j \leq n)$ be objective profiles
over Muller objectives.
Deciding whether there exists a $(\bmalpha_1,\ldots,\bmalpha_n)$-NE
is decidable.}

\begin{proof}
	By Corollary \ref{cor:n-nash}, there exists a $(\bmalpha_1,\ldots,\bmalpha_n)$-NE
	if and only if there exists a play $\rho=v_0v_1v_2\cdots\in\Play$ satisfying Condition (\ref{con:n-nash}).
	\begin{algorithm}[tb]
		\caption{}
		\label{alg:n-nash}
		\begin{algorithmic}[1]
			\Require a game arena $\G=\Arena$ and objective profiles $\bmalpha_j=(O_p^j)_{p\in P} \ (1\leq j\leq n)$.
			\ForAll {$v\in V$}
				\State Let $p\in P$ be the player such that $v \in V_p$.
				\State $O_v:=\bigcap_{O_p^j\in\Winnable^p_{(\G,v)},1\leq j\leq n}O_p^j$.
			\EndFor
			\State Nondeterministically select a set of vertices $V'\subseteq V$ and construct a $1$-player subgame arena $\G_{V'}=(\{1\},V',(V'),v_0,E')$ of $\G$.
			\State $O_{\G_{V'}}:=\bigcap_{v \in V'}O_v$.
			\If{Player $1$ has a winning strategy $\sigma_1\in\Sigma^1_{\G_{V'}}$ for $\G_{V'}$ and $O_{\G_{V'}}$}
				\State \Return Yes with $\sigma_1$
			\Else
				\State \Return No
			\EndIf
		\end{algorithmic}
	\end{algorithm}
	Algorithm \ref{alg:n-nash} decides the existence of a play satisfying Condition (\ref{con:n-nash}).
	In Algorithm \ref{alg:n-nash}, we call a game arena 
	$\calG_{V'} = (\{1\}, V', (V'), v_0, E')$ satisfying $V' \subseteq V, v_0 \in V'$ and $E'=\{ (v,v') \in E \mid v,v'\in V'\}$
	a 1-player subgame arena of $\calG$ (induced by $V'$).

	Let us show the correctness of Algorithm \ref{alg:n-nash}.
	First, we show that when Algorithm \ref{alg:n-nash} answers Yes, 
	the outcome of the strategy answered by Algorithm \ref{alg:n-nash}
	satisfies Condition (\ref{con:n-nash}).
	Let $\rho=\out_{\G_{V'}}(\sigma_1) = v_0v_1v_2 \cdots \in \Play$
	for the strategy $\sigma_1$ returned by Algorithm \ref{alg:n-nash}.
	Because $\rho$ is the outcome of a winning strategy for $O_{\G_{V'}}$, we have $\rho \in O_{\G_{V'}}$.
	By the definitions of $O_{\G_{V'}}$ and $O_v$,
	\begin{equation*}
		\begin{aligned}
			\rho\in O_{\G_{V'}} \iff &\forall v\in V'.\ \rho\in O_v \\
			\iff &\forall v\in V'.\ \forall p \in P.\ \forall 1\leq j\leq n.\ \\
			&(v\in V_p \wedge O_p^j\in\Winnable^p_{(\G,v)})\Rightarrow \rho\in O_p^j.
		\end{aligned}
	\end{equation*}
	Because $\rho$ is a play in $\G_{V'}$, we have $v_i\in V'$ for all $i\geq 0$.
	Thus, 
	\begin{equation*}
		\begin{aligned}
			\rho\in O_{\G_{V'}} \Rightarrow &\forall p\in P.\ \forall i\geq 0.\ \forall 1\leq j\leq n.\ \\
			&(v_i\in V_p \wedge O_p^j\in\Winnable^p_{(\G,v_i)}\Rightarrow \rho\in O_p^j).
		\end{aligned}
	\end{equation*}
	Because $O^j_p$ is prefix-independent, $\rho\in O_p^j \iff v_iv_{i+1}v_{i+2}\cdots\in O_p^j$.
	Therefore $\rho$ satisfies Condition (\ref{con:n-nash}).
	Conversely, we show that 
	if there exists a play $\rho$ satisfying Condition (\ref{con:n-nash}), 
	then at least one nondeterministic branch of Algorithm \ref{alg:n-nash} should answer Yes 
	with a strategy $\sigma_1$ such that $\rho=\out_{\G_{V'}}(\sigma_1)$.
	Assume that there exists a play $\rho=v_0v_1v_2\cdots\in\Play$ satisfying Condition (\ref{con:n-nash}).
	Let $V'=\{v \in V \mid \exists i\geq0.\ v=v_i\}$,
	and then construct the $1$-player subgame arena $\G_{V'}=(\{1\},V',(V'),v_0,E')$
	with $E' = \{ (v,v') \in E \mid v,v'\in V'\}$
	and the objective $O_{\G_{V'}}=\bigcap_{v\in V'}O_v$ where for all $v\in V'$,
	$O_v = \bigcap_{O^j_p\in\Winnable^p_{(\G,v)},1\leq j\leq n}O_p^j$ for $p\in P$ such that $v\in V_p$.
	It is easy to see that $\rho$ is a play of $\G_{V'}$ and $\rho\in O_{\G_{V'}}$ by Condition (\ref{con:n-nash}).
	Therefore, any strategy $\sigma_1$ that produces $\rho$ is a winning strategy 
	of the player $1$ for $\G_{V'}$ and $O_{\G_{V'}}$,
	and Algorithm \ref{alg:n-nash} should answer Yes with 
	a strategy $\sigma_1$ such that $\rho=\out_{\G_{V'}}(\sigma_1)$.
	\qed
\end{proof}

\end{document}